\documentclass[%
 reprint,
superscriptaddress,
 amsmath,
 amssymb,
 aps,
prb,
]{revtex4-1}

\usepackage[dvipdfmx]{graphicx}
\usepackage{bm}
\usepackage{overpic}
\usepackage{wrapfig}
\usepackage{color}
\usepackage{amsthm}
\newtheorem{thm}{Theorem}

\newcommand{\red}[1]{\textcolor{red}{#1}}

\begin{document}


\title{Exact solutions, spectrum properties, and  hierarchical structures of the multiple temperature model }

\author{Hiroki Katow}
\email{hkatow@atto.t.u-tokyo.ac.jp}
\affiliation{%
 Photon Science Center, Graduate School of Engineering, The University of Tokyo, 7-3-1 Hongo, Bunkyo-ku, Tokyo 113-8656, Japan
}

\author{Kenichi L. Ishikawa}
\email{ishiken@n.t.u-tokyo.ac.jp}
\affiliation{%
 Photon Science Center, Graduate School of Engineering, The University of Tokyo, 7-3-1 Hongo, Bunkyo-ku, Tokyo 113-8656, Japan
}
\affiliation{%
 Department of Nuclear Engineering and Management, Graduate School of Engineering, The University of Tokyo,7-3-1 Hongo, Bunkyo-ku, Tokyo 113-8656, Japan
}%
\affiliation{%
 Research Institute for Photon Science and Laser Technology\, The University of Tokyo, 7-3-1 Hongo, Bunkyo-ku, Tokyo 113-0033, Japan
}
%
%

\date{\today}

\begin{abstract}
Recent developments of ultrafast laser pulse techniques enable us to study the subpicosecond scale dynamics out of thermal equilibrium.
Multiple temperature models are frequently used to describe such  {dynamics} where the total system is divided into subsystems  {each of which is} in local thermal equilibrium.
 {Typical examples include}
the electron-lattice two temperature model and electron-spin-phonon three temperature model.
We  {present the exact analytical solutions of linear multiple temperature model  {(MTM)}, based on the Fourier series expansion, and discuss their properties for the case of the two and three temperature models.} 
 {We show that the general solution of MTM is expressed as a linear combinations of a spatially uniform, single-temperature stationary mode and the other non-oscillatory, decaying ``eigenmodes" characterized by different wave vectors and well-defined mode lifetimes.}
The eigenmode picture enables us to  {explore} the hierarchical structure of models with  {respect} to space, time and the coupling parameter.
We also find diffusion modes  {unique to} the three temperature model which unveils the rich physics in spite of the simplicity of the model. 
 {Furthermore, we prove that the general linear multiple temperature model fulfills physical requirements such as energy conservation and convergence to a spatially uniform, single-temperature steady state with well-defined mode lifetimes.}
\end{abstract}

\maketitle
\section{Introduction}
 The nonequilibrium description of the condensed matter systems has remained  {a subject of strong interest} for decades in physics.
 Leaving from well-defined thermodynamical equilibrium states, a possible first step towards  {the description of nonequilibrium dynamics} is to divide the total system into subsystems,  {each} in local thermal equilibrium.
 The idea of separating the total system into electronic and lattice subsystems with different temperature dates back to 1950s \cite{KLT57}.
 Early developments of this idea is detailed in a review by Kabanov \cite{K20}.
 A present form of the two temperature model (2TM) can be found in early 1970s\cite{AKP74}.
 A theoretical proposal to measure the electron-phonon coupling strength by pump-probe experiments\cite{A87} followed by observations in superconducting metallic systems\cite{BKMFCID90,CFGLLMYS91} has paved the way to a crucial application of the 2TM.    
 
The 2TM is now applied to extreme conditions where melting, evaporation, and material removal occur by ultrafast laser excitation\cite{CMNvAT96,NMJTCWW97},  {e.g., during the ultrafast laser material processing, for which higher energy efficiency and spatial precision are expected.}
Subpicosecond laser pulse deposits energy on the electronic subsystem in a ultrashort time scale while the lattice temperature remains relatively low.
The fast thermalization process of the electronic system is considered to justify that the electronic and lattice system possess different temperatures $T_e$ and $T_l$ after the laser pulse is turned off.  

The limitation of the 2TM has been recognized early on. 
Its failures of predicting the electron-phonon relaxation time at low temperature and its excitation intensity dependence were pointed out in \cite{GSL92,GSL95}.
Baranov and Kabanov derived a temperature range $\hbar^2\omega_D^2/E_F < k_BT < \hbar\omega_D(E_F/\hbar\omega_D)^{1/3}$ where 2TM cannot be justified\cite{BK14}. $\omega_D$ is the Debye frequency, and $E_F$ is the Fermi energy.
The Boltzmann equation approach is frequently used to improve the description of nonthermal distribution function\cite{FVATCV00,PCLC04,PCLC07,KA08,MR13}.
For the description of material destruction processes, a  {multi-scale} modeling which combines the 2TM and the classical molecular dynamics is employed\cite{IROGVZ08,IKLRCGS13,IBISGR17}. 
A recent review can be referred for this approach\cite{RIGA17}.   

Yet simple but a straightforward extension of the 2TM is dividing the system into smaller subsystems.
Waldecker introduced an idea to generalize the 2TM to the nonthermal lattice model where three phonon branches of Al have their own temperatures \cite{WBE16}.
A similar approach is  {applied to} graphene \cite{BOMMSR20}.
The electron-spin-phonon three temperature model has been developed to explain the ultrafast demagnetization process \cite{BMDB96,KKWHAC14,ZJWS21}, sometimes in combination with a microscopic equation of motion\cite{KKWHAC14,ZJWS21}.

In this paper we present exact solutions of the linear multiple temperature model (MTM) whose coefficients are all constant. 
Under the condition of vanishing heat flow of each subsystem at the boundaries, the model can be diagonalized.
The system dynamics can then be described by a linear combination of damping eigenmodes.
Each eigenmode is characterized by the mode lifetime which depends on the wave vector $\mathbf{q}$.
We firstly discuss the 2TM.
The exact solution of the 2TM splits into two eigenmodes whose eigenvalues form two branches $\zeta_\pm(\mathbf{q})$ on $\mathbf{q}$ space.
We will see that the $\zeta_+(\mathbf{q})$ branch is smoothly connected to the solution of an effective one temperature model (1TM) in a small $\mathbf{q}$ limit.
While in an opposite, large $\mathbf{q}$ limit, $\zeta_\pm(\mathbf{q})$ branches converge to the free diffusion modes of electron and lattice temperatures where the electron-lattice coupling $G$ becomes negligible.
The eigenmode picture thus enables us to explore the spatial scale dependence of the model behaviors. 
The result of a case for gold highlights this point.
we then provide the exact solution of linear three temperature model (3TM) which consists of three subsystems.
The additional degree of freedom leads to an emergence of a special solution which does not have any counter part to the 2TM solutions.
This solution can physically be interpreted as pure phonon-phonon, or spin-phonon diffusion modes where the amplitude of electron temperature is completely suppressed.
We also investigate a ``weak coupling limit" $G_{12}, G_{13}\ll G_{23}$ of 3TM, where $G_{12}$, $G_{13}$, and $G_{23}$ are coupling parameters between subsystems.
We show the 3TM can be approximated by an effective 2TM in the weak coupling limit combined with small $\mathbf{q}$ limit.
This result clarifies a hierarchical structure of the MTM. 
The 3TM includes 2TM, and 2TM includes 1TM in appropriate limits of spacial, time and parameter scale.
We finally derive a series of theorems which strongly restrict the eigenvalue properties of the MTM.
According to the theorems the MTM eigenvalue is always non-positive real valued, which assures the mode lifetime is always well defined.
A stationary solution is also guaranteed to present only in $\mathbf{q}=0$ point.
The results of this paper will serve as a foundation of advanced models, {\it e.g.}, with non-linearity or spatial non-uniformity.

This paper is organized as follows.
In Sec.~\ref{2TM} we introduce the 2TM and derive its exact solutions.
We will examine the 2TM behaviors for the case of bulk gold.
In Sec.~\ref{3TM} we introduce the 3TM and derive its exact solutions.
The ``band" structure of eigenvalues on three-dimensional parameter space helps us to grasp an overview of solutions.
The derivation of effective 2TM from the 3TM is also discussed.
In Sec.~\ref{theorem} we derive a series of theorems which provide strong and physically reasonable limitations on the MTM eigenvalue properties.

\section{Linear two temperature model}
\label{2TM}
The linear two-temperature model is defined as,
\begin{eqnarray}
&&
\begin{bmatrix}
C_e && 0 \\
0 && C_l \\
\end{bmatrix}
\frac{\partial}{\partial t}
\begin{bmatrix}
T_e (t,\mathbf{r}) \\
T_l (t,\mathbf{r}) \\
\end{bmatrix}
\nonumber
\\ 
&&=
\left\{
\begin{bmatrix}
\kappa_e\nabla^2 & 0\\
0 & \kappa_l\nabla^2\\
\end{bmatrix}
+
\begin{bmatrix}
-G & G\\
G & -G\\
\end{bmatrix}
\right\}
\begin{bmatrix}
T_e (t,\mathbf{r}) \\
T_l (t,\mathbf{r})\\
\end{bmatrix}.
\label{eq:2TM}
\end{eqnarray}
 {Here $T_e (t,\mathbf{r})$ and $T_l (t,\mathbf{r})$ are the electron and lattice temperature at position $\mathbf{r}$ and time $t$.} 
$C_e$ ($C_l$) and $\kappa_e$ ($\kappa_l$) denote the heat capacity and the thermal diffusion coefficient  of the electronic (lattice) subsystem, respectively, and
$G$ the electron-lattice coupling constant.
Parameters $C_e$, $C_l$, $\kappa_e$, $\kappa_l$, and $G$ are all positive real valued.
Throughout this paper we 
 {assume that the system is rectangular shaped whose side lengths are given by $L_i$ ($i=x,y,z$), and}
use a boundary condition
\begin{eqnarray}
 \partial_{r_i} T_e = 0,\; \partial_{r_i}  T_l=0  {\rm \; for\;}r_i=0,L_i (i=x,y,z).
 \label{eq:BC}
\end{eqnarray}
Clearly $\prod_i\cos (q_{n_i}r_i)$ $(i=x,y,z)$ is an eigenfunction of the diffusion term, where the wave vector is defined by $\mathbf{q}_{n_xn_yn_z}=(q_{n_x},q_{n_y},q_{n_z})=(\pi n_x/L_x,\pi n_y/L_y,\pi n_z/L_z)$, with
$n_i$ $(i=x,y,z)$ being a non-negative integer.
For simplicity we omit the subscript from the wave vector hereafter. 
 {Thus, the general solution of Eq.~(\ref{eq:2TM}) can be expressed by a linear combination of different wave vector components of the form 
\begin{eqnarray}
\label{eq:FSE}
\begin{bmatrix}
T_e(t,\mathbf{r})\\
T_l(t,\mathbf{r})
\end{bmatrix}
=\sum_{\mathbf{q}}
\begin{bmatrix}
A_{\mathbf{q}}\\
B_{\mathbf{q}}\\
\end{bmatrix}
\prod_i\cos(q_ir_i)e^{\zeta({\bf q})t},
\end{eqnarray}
as a natural extension of the Fourier series expansion common in the studies of thermal diffusion \cite{NMJTCWW97,OIINNH16,YA20}.
Because of the spatial uniformity, each $\mathbf{q}$ component is independent.
Then, the coefficients $\begin{bmatrix}
A_{\mathbf{q}}\\
B_{\mathbf{q}}\\
\end{bmatrix}$ and $\zeta ({\bf q})$ are the eigenvectors and eigenvalues, respectively, of}
a $2\times2$ non-symmetric matrix,
\begin{eqnarray}
H' =
\begin{bmatrix}
\omega_e(\mathbf{q}) -\Omega_e& \Omega_e\\
\Omega_l & \omega_l(\mathbf{q}) -\Omega_l\\
\end{bmatrix}
\end{eqnarray}
where $\omega_e(\mathbf{q})=-\frac{\kappa_e}{C_e}q^2$, $\omega_l(\mathbf{q})=-\frac{\kappa_l}{C_l}q^2$, with $q=|{\bf q}|$, and $\Omega_e=G/C_e,\; \Omega_l=G/C_l$.
 {It is interesting to notice that the analytical form of $H'$ is analogous to the Hamiltonian of other physical systems such as the quantum Rabi model and the polariton model except $H'$ is not symmetric}. 
The eigenvalue $\zeta$ of $H'$ splits into the upper and lower branches $\zeta(\mathbf{q})=\zeta_+(\mathbf{q})$ and $\zeta_-(\mathbf{q})$, respectively: 
\begin{eqnarray}
\zeta_\pm(\mathbf{q})=\Delta_+(\mathbf{q})/2\pm\sqrt{\{\Delta_-(\mathbf{q})/2\}^2+\Omega_e\Omega_l}
\label{eq:zeta_2TM}
\end{eqnarray}
where
\begin{eqnarray}
\Delta_\pm(\mathbf{q}) = \{\omega_l(\mathbf{q})-\Omega_l\}\pm\{\omega_e(\mathbf{q})-\Omega_e\}.
\end{eqnarray}
The corresponding right eigenmode (eigenvector) is given by  
\begin{eqnarray}
\mathbf{v}^R_{\zeta\mathbf{q}}(t,\mathbf{r})=
\begin{bmatrix}
v^e_{\zeta\mathbf{q}}(t,\mathbf{r})\\
v^l_{\zeta\mathbf{q}}(t,\mathbf{r})
\end{bmatrix}
=
A_{\zeta\mathbf{q}}
\begin{bmatrix}
1\\
R_{\zeta\mathbf{q}}/\Omega_e
\end{bmatrix}
u_{\zeta\mathbf{q}}(t,\mathbf{r})
\label{eq:EV_2TM}
\end{eqnarray}
where
\begin{eqnarray}
&&u_{\zeta\mathbf{q}}(t,\mathbf{r}) = \prod_i\cos(q_ir_i)e^{\zeta(\mathbf{q})t}\\
&& R_{\zeta\mathbf{q}}=\Delta_-(\mathbf{q})/2\pm\sqrt{\{\Delta_-(\mathbf{q})/2\}^2+\Omega_e\Omega_l}.
\end{eqnarray}
 {The general solution Eq.~\eqref{eq:FSE} of Eq.~(\ref{eq:2TM}) is given by,
\begin{equation}
    \begin{bmatrix}
    T_e(t,\mathbf{r})\\
    T_l(t,\mathbf{r})
    \end{bmatrix}
    =\sum_{\mathbf{q}\,(n_x,n_y,n_z)}
    \left[\mathbf{v}^R_{\zeta_+\mathbf{q}}(t,\mathbf{r})+\mathbf{v}^R_{\zeta_-\mathbf{q}}(t,\mathbf{r})\right].
\end{equation}
}
 {Since the linear temperature model assumes a  spatially uniform system, different wave vector components do not couple with each other.}
The mode amplitude $A_{\zeta\mathbf{q}}$ is determined by the initial condition.

 {It follows from $\omega_e({\bf q}=0)=\omega_l({\bf q}=0)=0$ that, 
\begin{equation}
\label{eq:zeta+0}
    \zeta_+({\bf q}=0) = 0,
\end{equation}
and,
\begin{equation}
\label{eq:zeta-0}
    \zeta_-({\bf q}=0) = -(\Omega_e+\Omega_l)<0.
\end{equation}
One can show that $\partial_{\omega_e}\zeta_\pm > 0$ and $\partial_{\omega_l}\zeta_\pm > 0$, therefore, $\zeta_\pm$ monotonically decreases with increasing $q$ [see Fig.~\ref{fig:spectrum}(a) below].
Furthermore, $\zeta_+({\bf q}=0) = 0$ [Eq.~\eqref{eq:zeta+0}], and, otherwise, $\zeta_\pm < 0$, indicating that all the modes damp except for ${\bf v}^R_{\zeta_+0}$, which corresponds to the final state; the larger the wave number, the faster the mode damps on each branch.
It should also be noted that each individual eigenmode except for ${\bf v}^R_{\zeta_+0}$ cannot be a physical solution alone, since it spatially oscillates around zero.
The general solution must be a superposition of two or more modes to ensure non-negative temperature everywhere in the system.
}

Asymptotic behaviors of the solution in small and large $\bf{q}$ limit are informative to see the nature of this model.
For $\mathbf{q}\rightarrow 0$ limit,
\begin{eqnarray}
\zeta_+(\mathbf{q})\rightarrow -\frac{\kappa_e+\kappa_l}{C_e+C_l}q^2
\label{eq:upper_l1small}
\end{eqnarray}
\begin{eqnarray}
R_{\zeta_+\mathbf{q}}/\Omega_e \rightarrow 1-\frac{1}{\Omega_e+\Omega_l}(\kappa_l/C_l-\kappa_e/C_e)q^2 
\label{eq:upper_R1small}
\end{eqnarray}
and 
\begin{eqnarray}
\zeta_-(\mathbf{q})\rightarrow -(\Omega_e+\Omega_l)
-\frac{1}{C_e+C_l}\left( \frac{C_l}{C_e} \kappa_e+\frac{C_e}{C_l}\kappa_l\right)q^2\nonumber \\
\label{eq:lower_l1small}
\end{eqnarray}
\begin{eqnarray}
R_{\zeta_-{\mathbf q}}/\Omega_e\rightarrow -\frac{C_e}{C_l}+\frac{C_e}{C_l}\frac{1}{\Omega_e+\Omega_l}\left(\frac{\kappa_l}{C_l}-\frac{\kappa_e}{C_e}\right)q^2.
\label{lower_R1small}
\end{eqnarray}
We see that the upper branch $\zeta_+$ reduces to an effective ``one temperature model" with the effective heat capacity $C_{\rm eff}=C_e+C_l$ and the effective thermal diffusion coefficient $\kappa_{\rm eff}=\kappa_e+\kappa_l$. 
Up to the $q^2$ order, only the relative amplitude Eq.~(\ref{eq:upper_R1small})  {between the lattice and electronic systems} provides the information of the electron-lattice coupling $G$. 
Contrary to the upper branch, the lifetime of the lower branch $-1/\zeta_-(\mathbf{q}=0)=1/(\Omega_e+\Omega_l)$ enables us to determine the value of $G$ in its leading term. 

Next we examine the large $\mathbf{q}$ limit, which corresponds to $\omega_e(\mathbf{q}),\omega_l({\mathbf{q}})\gg \Omega_e,\Omega_l$.
For the upper branch:
\begin{eqnarray}
\zeta_+(\mathbf{q})\rightarrow -\min \left(\frac{\kappa_e}{C_e}, \frac{\kappa_l}{C_l}\right)q^2
\label{eq:upper_l1large}
\end{eqnarray}
\begin{eqnarray}
R_{\zeta_+\mathbf{q}}/\Omega_e \rightarrow 
\left\{
\begin{array}{l}
\left(\frac{\kappa_l}{C_l}-\frac{\kappa_e}{C_e}\right)q^2 \;{\rm for \;}\frac{\kappa_e}{C_e}< \frac{\kappa_l}{C_l} \\ 
0  \;{\rm for \;}\frac{\kappa_e}{C_e}> \frac{\kappa_l}{C_l}
\end{array}
\right.
\label{eq:upper_R1large}
\end{eqnarray}
and for the lower branch:
\begin{eqnarray}
\zeta_-(\mathbf{q})\rightarrow -\max \left(\frac{\kappa_e}{C_e}, \frac{\kappa_l}{C_l}\right)q^2
\label{eq:lower_l1large}
\end{eqnarray}
\begin{eqnarray}
R_{\zeta_-\mathbf{q}}/\Omega_e \rightarrow 
\left\{
\begin{array}{l}
0  \;{\rm for \;}\frac{\kappa_e}{C_e}< \frac{\kappa_l}{C_l}\\
\left(\frac{\kappa_l}{C_l}-\frac{\kappa_e}{C_e}\right)q^2 \;{\rm for \;}\frac{\kappa_e}{C_e}> \frac{\kappa_l}{C_l} 
\end{array}
\right..
\label{eq:lower_R1large}
\end{eqnarray}
In this limit, the electron-lattice coupling $G$ is negligible, and the system dynamics is dominated by ``free diffusion process".

\begin{figure*}
\begin{tabular}{ccc}
\begin{minipage}{0.325\hsize}
\begin{center}
\begin{overpic}[angle=0,width=5.8cm]{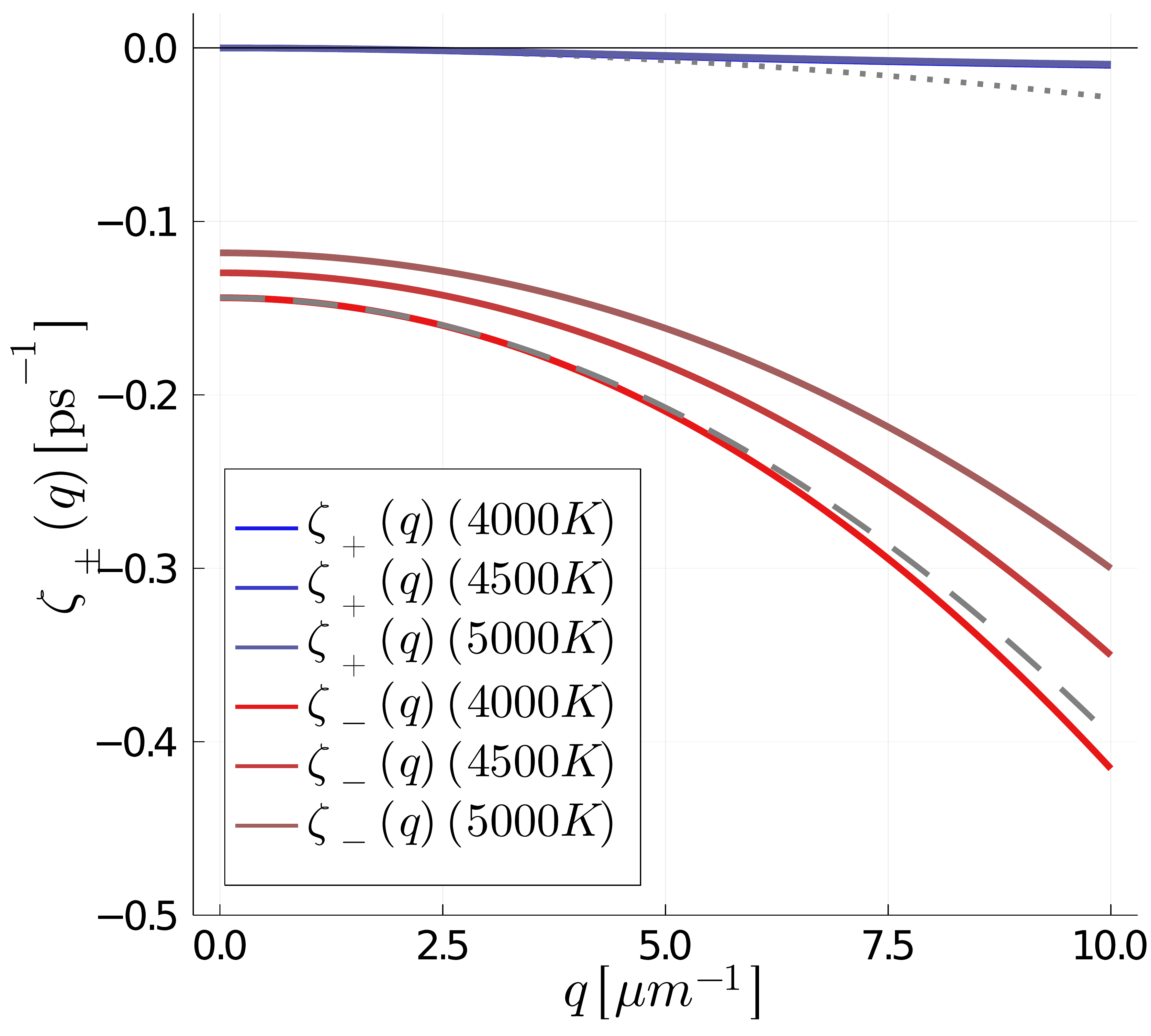}
\put(0,80){(a)}
\end{overpic}
\end{center}
\end{minipage}
\begin{minipage}{0.325\hsize}
\begin{center}
\begin{overpic}[width=5.8cm]{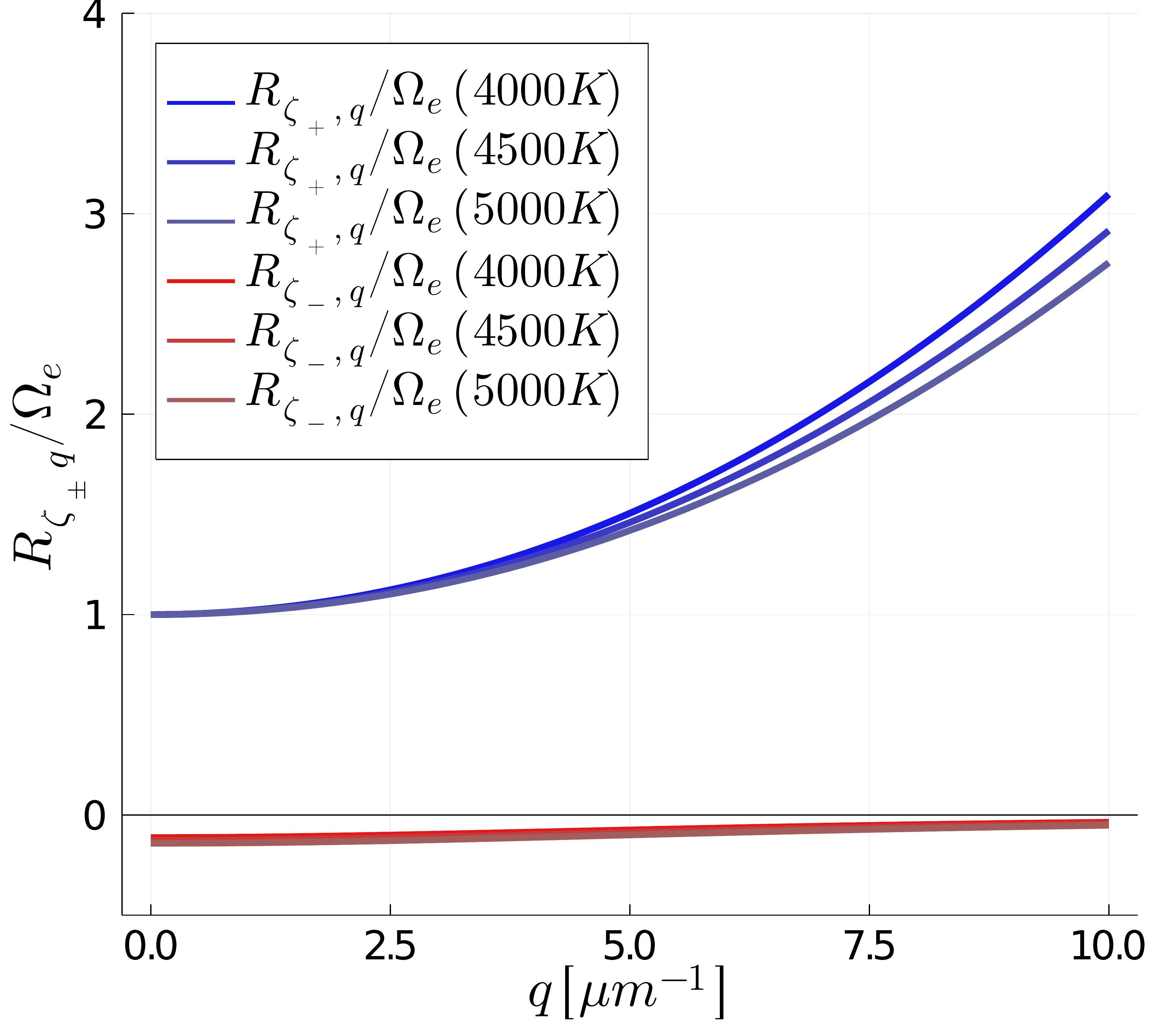}
\put(0,80){(b)}
\end{overpic}
\end{center}
\end{minipage}
\begin{minipage}{0.325\hsize}
\begin{center}
\begin{overpic}[width=5.8cm]{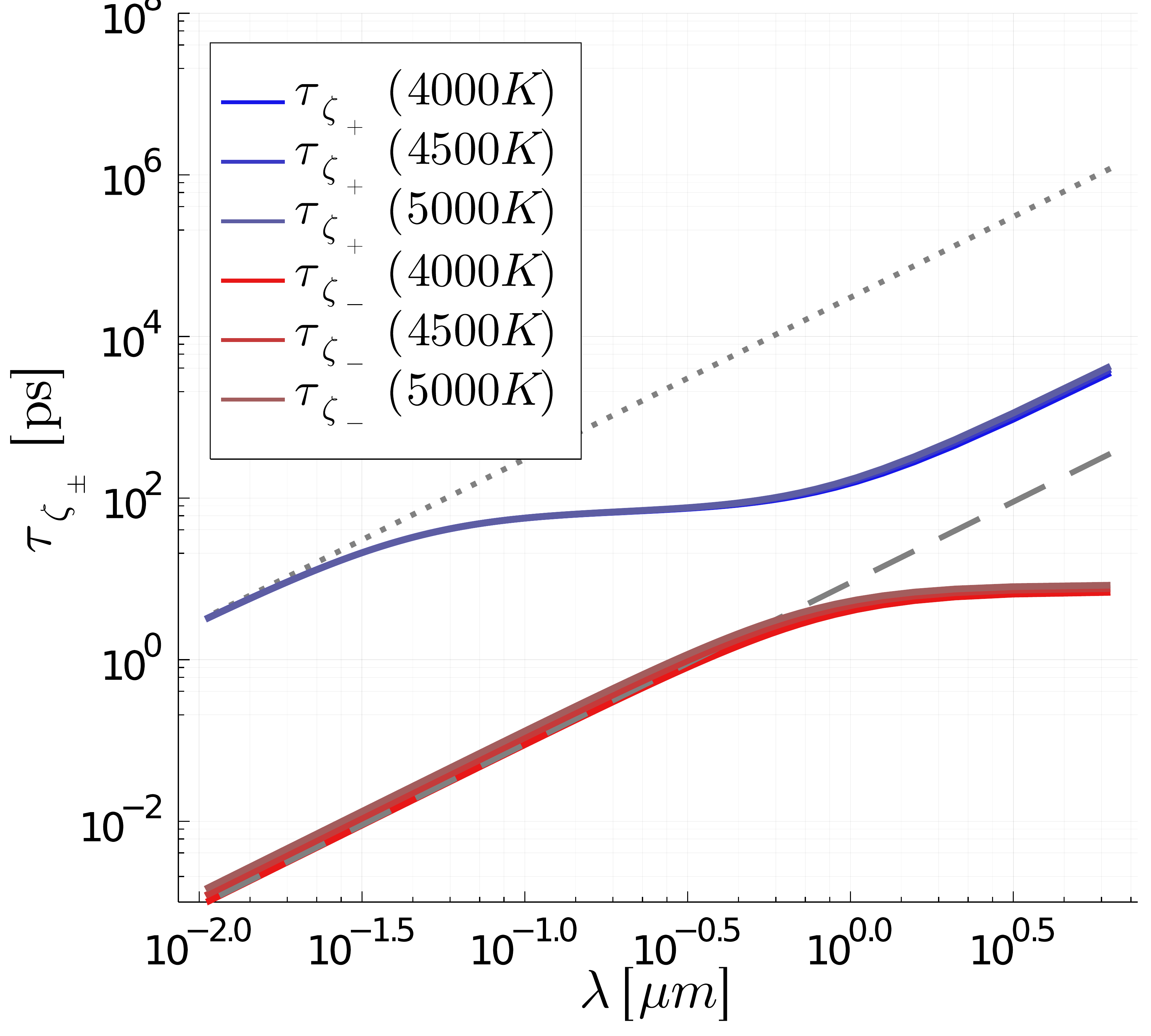}
\put(0,80){(c)}
\end{overpic}
\end{center}
\end{minipage}
\end{tabular}
\caption{Wave vector $q$ dependence of (a) the eigenvalue of two-temperature model Eq.~(\ref{eq:zeta_2TM}), (b) the relative amplitude of lattice temperature given in Eq.~(\ref{eq:EV_2TM}). $\zeta_\pm$ in small $|q|$ limit Eq.~(\ref{eq:upper_l1small}) and Eq.~(\ref{eq:lower_l1small}) for $T_e=4000K$ are shown by dotted and dashed lines, respectively.
(c) Wave length $\lambda=2\pi/q$ dependence of the lifetime $\tau_{\zeta_\pm}=-1/\zeta_\pm$, where small $\lambda$ limit Eq.~(\ref{eq:lower_l1large}) and Eq.~(\ref{eq:upper_l1large})  for $T_e=4000K$ are shown by dotted and dashed lines, respectively.
We used parameters of gold given by \cite{RIGA17,WRLD94, JM16} for fixed electron temperature $T_e=4000,4500,5000K$. } 
\label{fig:spectrum}
\end{figure*}

Now, as a specific example, let us investigate the behaviors of the modes for the case of gold.
We referred the values in literature as $C_e(T_e)=\gamma T_e $ where $\gamma=67.6 \; {\rm J/m^3K^2}$,
\begin{eqnarray}
\kappa_e(T_e,T_l)=\frac{1}{3}v_F^2 C_e(T_e)\frac{1}{AT_e^2+BT_l}
\end{eqnarray}
where $v_F=1.39\times10^6 \;{\rm ms^{-1}}$, $A=1.2\times 10^7 \; {\rm s^{-1}K^{-2}}$, $B=1.23\times10^{11} {\rm s^{-1}K^{-1}}$, and $G=3.5\times10^{16} \;{\rm J/m^3Ks}$ from \cite{RIGA17,WRLD94}.
The lattice heat capacity $C_l=2.4\times10^6$ and $\kappa_l=2 \;{\rm J/mKs}$ are taken from \cite{JM16}.
 {Although $C_e$ and $\kappa_e$ are, strictly speaking, temperature dependent, we focus on the behavior of the linear 2TM here and use the constant $C_e$ and $\kappa_e$ values calculated for $T_e = 4000, 4500, 5000$K.}
These belong to a typical temperature scale in the laser ablation processes\cite{RIGA17}.
In Fig.~\ref{fig:spectrum} (a) we can confirm that all eigenvalues $\zeta_\pm (q)$ are negative real valued, monotonically decreasing with $q$, and hence the lifetime of each mode $\tau_{\zeta_\pm}=-1/\zeta_\pm$ is well defined
 {except for $\mathbf{q}=0$ where $\zeta_+=0$}.
Once the initial condition is given, the system dynamics is completely described by the damping process of each mode.
We find that the asymptotic solution Eq.~(\ref{eq:upper_l1large}) reproduces 87\% of the exact value for wave length $\lambda=2\pi/q=0.5 {\rm \mu m}$ and  Eq.~(\ref{eq:lower_l1large}) gives 101\% for $\lambda=1.0 {\rm \mu m}$.  
The relative amplitude [Fig.~\ref{fig:spectrum} (b)] shows a qualitative difference between the upper and lower branch. 
In the upper branch ($\zeta_+$) the electron and lattice temperatures spatially oscillate in phase, while in the lower, or $\zeta_-$ branch the oscillation is antiphase.
Figure~\ref{fig:spectrum} (b) also shows that in the large $\mathbf{q}$ limit the amplitude of electron (lattice) temperature in the upper (lower) branch vanishes, which indicates a transition to the free diffusion process. 
We can also see this transition in  {Fig.~\ref{fig:spectrum} (c), which plots 
the $\lambda$ dependence of lifetime $\tau_{\zeta_\pm}$};  
the exact solutions Eq.~(\ref{eq:zeta_2TM}) approach to the asymptotic solutions Eqs.~(\ref{eq:upper_l1large}) and (\ref{eq:lower_l1large}) in the small $\lambda$, i.e., large $\mathbf{q}$, limit.
Figure~\ref{fig:spectrum} (c) indicates that such a transition occurs at tenth of nanometer scales in the upper branch and at sub $\mu$m scale in the lower branch.

\section{Linear three temperature model}
\label{3TM}
As a natural extension of the 2TM, the three temperature model (3TM) is defined as follows:
\begin{eqnarray}
\Lambda \frac{\partial}{\partial t}\mathbf{T}=H \mathbf{T}.
\label{eq:3TM}
\end{eqnarray}
Here $\mathbf{T}$ is a three component vector, representing the temperatures of the three subsystems.
$\Lambda$ and $H$ are symmetric $3\times 3$ matrices given by,
\begin{eqnarray}
\Lambda_{ij}=C_i\delta_{ij}
\end{eqnarray}
\begin{eqnarray}
H_{ii}&=&\kappa_i\nabla^2-\sum_{k\neq i}^3G_{ik}\\
H_{ij}&=&H_{ji}=G_{ij} \textrm{ \;for \; }(i\neq j) \label{eq:3TMHij}
\end{eqnarray}
We note again that the heat capacity $C_i$, thermal diffusion coefficient $\kappa_i$, and the coupling constant between subsystems $G_{ij}$ are all positive real valued.
The matrices $\Lambda$ and $H$ are thus both real valued and symmetric $3\times 3$ matrices.

 Let us seek for the solution ${\bf T}$ of Eq.~\eqref{eq:3TM} expressed as a linear combination of different modes similar to Eq.~\eqref{eq:FSE}.
Then,  {we find three branches $\zeta_1 ({\bf q})$, $\zeta_2 ({\bf q})$, and $\zeta_3 ({\bf q})$ of eigenvalues of a non-symmetric matrix $H'=\Lambda^{-1}H$, 
by using the formula for the roots of the general cubic equation, as,} 
\begin{eqnarray}
\zeta_1(\mathbf{q})&=&-\{t^{1/3}_+(\mathbf{q})+t^{1/3}_-(\mathbf{q})\}+\alpha(\mathbf{q})/3 \label{eq:zeta_3TM1}\\
\zeta_2(\mathbf{q})&=&-\{\sigma^2t^{1/3}_+(\mathbf{q})+\sigma t^{1/3}_-(\mathbf{q})\}+\alpha(\mathbf{q})/3 \label{eq:zeta_3TM2}\\
\zeta_3(\mathbf{q})&=&-\{\sigma t^{1/3}_+(\mathbf{q})+\sigma^2t^{1/3}_-(\mathbf{q})\}+\alpha(\mathbf{q})/3 \label{eq:zeta_3TM3}
\end{eqnarray}
where 
\begin{eqnarray}
\sigma&=&e^{2\pi i/3}\\
t_{\pm}(\mathbf{q})&=&p_1(\mathbf{q})/2\pm\sqrt{p_1(\mathbf{q})^2/4+p_2(\mathbf{q})^3/27}\\
p_1(\mathbf{q})&=&-(2/27)\alpha(\mathbf{q})^3+(1/3)\alpha(\mathbf{q})\beta(\mathbf{q})+\gamma(\mathbf{q})\\
p_2(\mathbf{q})&=&\beta(\mathbf{q})-\alpha(\mathbf{q})^2/3
\end{eqnarray}
\begin{eqnarray}
\alpha(\mathbf{q})&=&\sum_i\Delta_i(\mathbf{q})\\
\beta(\mathbf{q})&=&\Delta_1(\mathbf{q})\Delta_2(\mathbf{q})+\Delta_2(\mathbf{q})\Delta_3(\mathbf{q})+\Delta_3(\mathbf{q})\Delta_1(\mathbf{q}) \nonumber \\
&&-\Omega_{12}\Omega_{21}-\Omega_{23}\Omega_{32}-\Omega_{13}\Omega_{31}\\
\gamma(\mathbf{q})&=&\Omega_{12}\Omega_{21}\Delta_3(\mathbf{q})+\Omega_{13}\Omega_{31}\Delta_2(\mathbf{q})+\Omega_{23}\Omega_{32}\Delta_1(\mathbf{q})\nonumber \\
&& -\Delta_1(\mathbf{q})\Delta_2(\mathbf{q})\Delta_3(\mathbf{q})-\Omega_{12}\Omega_{23}\Omega_{31}-\Omega_{13}\Omega_{32}\Omega_{21}\nonumber\\
\\
\Delta_i(\mathbf{q})&=&-\frac{\kappa_i}{C_i}\mathbf{q}^2-\sum_{j\neq i}\frac{G_{ij}}{C_i}=\omega_i(\mathbf{q})-\Omega_{ii}\\
 \Omega_{ij}&=&\frac{G_{ij}}{C_i}.
\end{eqnarray}

We show the global structure of the three branches Eqs.~(\ref{eq:zeta_3TM1})-(\ref{eq:zeta_3TM3}) on three dimensional parameter space spanned by $(\omega_1,\omega_2,\omega_3)$ in Fig.~\ref{fig:3TM}.
Figure \ref{fig:3TM}(a) is the exact solution Eq.~(\ref{eq:zeta_3TM1})-(\ref{eq:zeta_3TM3}).
We have chosen a path on the $(\omega_1,\omega_2,\omega_3)$ space to plot these solutions just like plotting the electronic band structure of periodic systems [Fig.~\ref{fig:3TM}(b)].
 {Note that} linear dispersion extending from $\Gamma=(0,0,0)$ point  {corresponds to} a parabolic band in the $\mathbf{q}$ space.
Values of $\mathbf{q}$ vectors on the path can be uniquely determined according to $\omega(\mathbf{q})=-(\kappa_i/C_i)q^2$ by specifying parameters $\kappa_i$ and $C_i$ of each subsystem.
Thus, Fig. \ref{fig:3TM}(a) show the global structure of the exact solution with various $\kappa_i$ and $\mathbf{q}$ values. 
Clearly the $\zeta_3$ branch Eq.~(\ref{eq:zeta_3TM3}) is the counterpart of $\zeta_+$ or the upper branch Eq.~(\ref{eq:zeta_2TM}) of the 2TM.
The lifetime of $\zeta_3$ branch diverges at $\Gamma$ point and the relative amplitude of all subsystems is always of the same sign as can be seen in Fig.~\ref{fig:3TM}(e).
On the other hand, the other two branches, $\zeta_1$ and $\zeta_2$, show behaviors unique to the three temperature system.
We can find such case on $\Gamma-$P100, $\Gamma-$P011, and some high symmetric axis.
On $\Gamma-$P100 axis, where $\kappa_2=\kappa_3=0$, $\zeta_2$ branch excludes the amplitude of subsystem 1 [Fig.~\ref{fig:3TM}(d)]. 
This approximation may apply to electron-longitudinal phonon-transverse phonon system, or electron-phonon-spin system.
The $\zeta_2$ branch then describes a purely phonon-like, or a purely spin-phonon diffusion mode which does not accompany electron thermal diffusion.  
In the same way $\Gamma-$P011 axis can be realized when $\kappa_1=0$ which may provide a good approximation of an electron-hole-phonon system  without nonlinearity. 
Then, The $\zeta_1$ branch indicates an electron-hole diffusion mode which does not accompany phonon thermal diffusion [Fig.~\ref{fig:3TM}(c)].
To our best knowledge these ``anomalous" diffusion modes have never been experimentally observed. 
Further investigations are required to clarify their role in the system dynamics.   
We point out that a similar solution known as the dark state can be found in the quantum three-level system driven by an external field  \cite{GZ15}.

\begin{figure*}[htbp]
\begin{center}
\includegraphics[width=17.5cm]{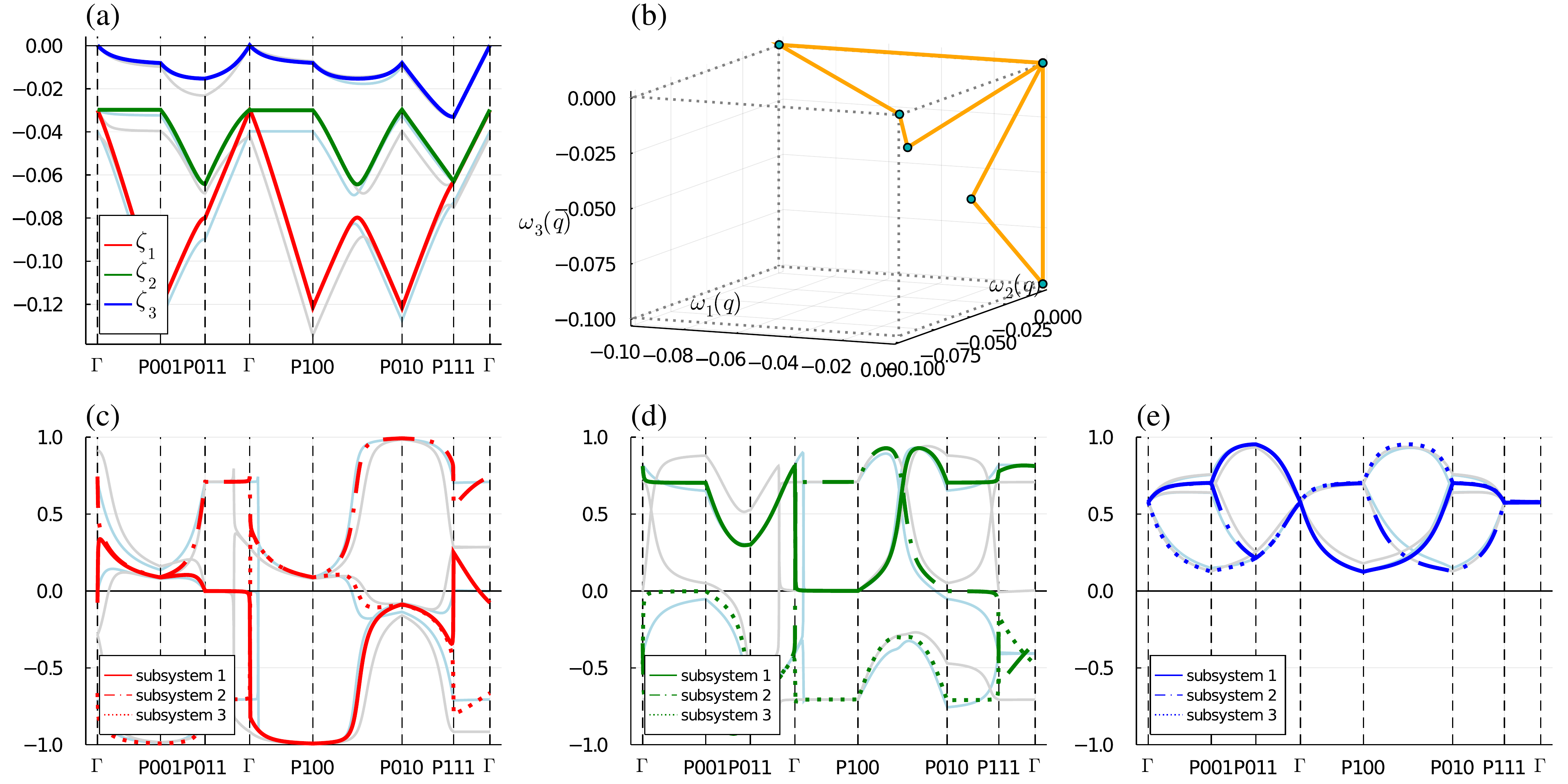}
\caption{Heat capacity $C_1=C_2=C_3=1.6 \times 10^6 {\rm J/m^3K}$ and the coupling between subsystems $G_{12}=G_{23}=G_{31}=1.6\times10^{16}{\rm J/sm^3K}$ are used to plot the exact solutions Eq.~(\ref{eq:zeta_3TM1}-\ref{eq:zeta_3TM3}) of the 3TM in (a).  
A path to plot the band structure (a) is shown in (b) on a parameter space spanned by $(\omega_1(\mathbf{q}), \omega_2(\mathbf{q}), \omega_3(\mathbf{q}))$. Coordinates of points are: $\Gamma=(0,0,0)$, P001=$(0,-0.1,0.0)$,P011=$(0,-0.05,-0.05)$, P100=$(-0.1,0.0,0.0)$, P010=$(0.0,-0.1,0.0)$, P111=$(-0.1,-0.1,-0.1)$. Numerical results of the relative amplitude of the each subsystem temperature of eigenmodes are shown for (c) $\zeta_1$, (d) $\zeta_2$, and (e) $\zeta_3$ branches.  The light blue line and light gray line in (a), (c)-(e) show the results when the parameter is changed as $C_1 \rightarrow 1.0\times 10^6 {\rm J/m^3K}$ and $G_{23}\rightarrow 2.4\times 10^{16}{\rm J/sm^3K}$, respectively.}
\label{fig:3TM}
\end{center}
\end{figure*}

We have found in the previous section that the ``effective one temperature model" is embedded in the linear 2TM.
Then, a question may naturally rise asking how an ``effective two temperature model" can be derived from the linear 3TM.
We can expect such solution will emerge when  {$\mathbf{q}$ is small} and two of the three subsystems are strongly coupled, i.e., $G_{23}\gg G_{12},G_{13}$.
Since the exact solution Eqs.~(\ref{eq:zeta_3TM1})-(\ref{eq:zeta_3TM3}) is too complicated to handle by a simple power expansion,
We put a start point on a weekly coupled 1+2 temperature model, where matrix $H^\prime$  is decomposed to,
\begin{eqnarray}
H'&=&\Lambda^{-1}H=H'_0+H'_1
\label{eq:Hprime}
\end{eqnarray}
\begin{eqnarray}
H'_0&=&
\begin{bmatrix}
\omega_1 & 0 & 0\\
0 & \omega_2-\Omega_{23} & \Omega_{23}\\
0 & \Omega_{32} & \omega_3-\Omega_{32}\\
\end{bmatrix}
\label{eq:H0}
\end{eqnarray}
\begin{eqnarray}
H'_1&=&
\begin{bmatrix}
-\Omega_{11} & \Omega_{12} & \Omega_{13}\\
\Omega_{21} & -\Omega_{21} & 0\\
\Omega_{31} & 0 & -\Omega_{31}\\
\end{bmatrix}
\label{eq:H1}
.
\end{eqnarray}
Equation (\ref{eq:H0}) is a block diagonal matrix describing a decoupled 1+2 temperature model,
whose eigenvalues $\zeta$ are simply given by $\zeta(\mathbf{q})=\zeta_1 (=\omega_1),\zeta_\pm$, where $\zeta_\pm$ is defined by Eq.~(\ref{eq:zeta_2TM}).
The two branches $\zeta_1(\mathbf{q})$ and $\zeta_+(\mathbf{q})$ are degenerate at $\mathbf{q}=0$ as $\zeta_1(\mathbf{q})=0$ and $\zeta_+(\mathbf{q})=0$.
 {This is in contrast to the spectrum in Fig.~\ref{fig:3TM}(a) for three subsystems  coupled with equal strength, where only $\zeta_3$ vanishes at $\mathbf{q}=0$ and $\zeta_1$ and $\zeta_2$ are degenerate there.}
As long as we restrict the timescale to $t \gg \tau_{\zeta_-}=-1/\zeta_-$, we can neglect the contribution of the $\zeta_-$ branch. 
Then, the corresponding right eigenvectors (eigenmodes) $\mathbf{v}^R_{\zeta_1}$, $\mathbf{v}^R_{\zeta_\pm}$ of $H_0^\prime$ are 
\begin{eqnarray}
\mathbf{v}^R_{\zeta_1}=
\begin{bmatrix}
1\\
0\\
0\\
\end{bmatrix},\;
\mathbf{v}^R_{\zeta_+}=
\begin{bmatrix}
0\\
1\\
R_{\zeta_+\mathbf{q}}/\Omega_{23}\\
\end{bmatrix},
\label{eq:vR}
\end{eqnarray}
and the left eigenvectors $\mathbf{v}^L_{\zeta_1}$, $\mathbf{v}^L_{\zeta_\pm}$ are
\begin{eqnarray}
\mathbf{v}^L_{\zeta_1}=
\begin{bmatrix}
1\\
0\\
0\\
\end{bmatrix},\;
\mathbf{v}^L_{\zeta_+}=\frac{1}{f_{\zeta_+\mathbf{q}}}
\begin{bmatrix}
0\\
1\\
R_{\zeta_+\mathbf{q}}/\Omega_{32}\\
\end{bmatrix}
\label{eq:vL}
\end{eqnarray}
where 
\begin{eqnarray}
R_{\zeta_+\mathbf{q}} =\zeta_+(\mathbf{q})-\{\omega_2(\mathbf{q})-\Omega_{23}\}\\
f_{\zeta_+\mathbf{q}}=1+R_{\zeta_+\mathbf{q}}^2/\Omega_{23}\Omega_{32}.
\end{eqnarray}
Here we have dropped the space- and time-dependent factors for simplicity.
Equations (\ref{eq:vR}) and (\ref{eq:vL}) satisfy the orthonormality relation
\begin{eqnarray}
{}^t\mathbf{v}_m^L\cdot \mathbf{v}_n^R=\delta_{mn}.
\end{eqnarray}
We introduce the new right eigenvector $\mathbf{w}^R_m$ of $H'$ in Eq.~(\ref{eq:Hprime}) by a linear combination of $\mathbf{v}^R_m$ as
\begin{eqnarray}
\mathbf{w}^R_m = \sum_{n=\zeta_1,\zeta_+} A_{mn}\mathbf{v}^R_n.
\end{eqnarray}
The amplitude $A_{mn}$, or transformation matrix, is determined by solving a following eigenvalue equation:
\begin{eqnarray}
K\mathbf{A}_m=\eta_m\mathbf{A}_m
\end{eqnarray}
where the elements of $2\times2$ matrix $K$ is given by
\begin{eqnarray}
K_{mn}={}^t\mathbf{v}^L_m \cdot H'\mathbf{v}^R_n
\end{eqnarray}
 {or, explicitly,}
\begin{eqnarray}
K_{\zeta_1\zeta_1}&=& \omega_1-\Omega_{11}\\
K_{\zeta_1\zeta_+}&=& \Omega_{12}+\frac{\Omega_{12}}{\Omega_{23}}R_{\zeta_+\mathbf{q}}\\
K_{\zeta_+\zeta_1}&=&f_{\zeta_+\mathbf{q}}^{-1}\left(\Omega_{21}+\frac{\Omega_{31}}{\Omega_{32}}R_{\zeta_+\mathbf{q}}\right) \\
K_{\zeta_+\zeta_+}&=& \zeta_+-f_{\zeta_+\mathbf{q}}^{-1}\left(\Omega_{21}+\frac{\Omega_{31}}{\Omega_{23}\Omega_{32}}R^2_{\zeta_+\mathbf{q}}\right).
\end{eqnarray}
$\mathbf{A}_m=(A_{m\zeta_1}, A_{m\zeta_+})$ and $\eta_m$ is an eigenvalue.
By taking a small $\mathbf{q}$ limit and omitting terms smaller than $O(G_{23}^{-1})$, we obtain 
\begin{eqnarray}
K\simeq H_{\rm eff}(\mathbf{q})+J(\mathbf{q})
\end{eqnarray}
where
\begin{eqnarray}
H'_{\rm eff}=
\begin{bmatrix}
-\frac{\kappa_1}{C_1}q^2-\frac{G_{12}+G_{13}}{C_1} & \frac{G_{12}+G_{13}}{C_1}\\
\frac{G_{12}+G_{13}}{C_2+C_3} & -\frac{\kappa_2+\kappa_3}{C_2+C_3} q^2-\frac{G_{12}+G_{13}}{C_2+C_3}
\end{bmatrix}\nonumber \\
\label{eq:Heff1}
\end{eqnarray}
and  {the matrix elements of $J(\mathbf{q})$ is given by,}
\begin{eqnarray}
J_{\zeta_1\zeta_1}&=&0\\
J_{\zeta_1\zeta_+}&=&-\frac{\Omega_{13}}{G_{23}}\frac{C_2\kappa_3-C_3\kappa_2}{C_2+C_3}q^2\\
J_{\zeta_+\zeta_1}&=&-\frac{\{(C_3-C_2)G_{13}+2C_3G_{12}\}(C_3\kappa_2-C_2\kappa_3)}{G_{23}(C_2+C_3)^3}q^2\nonumber\\
\\
J_{\zeta_+\zeta_+}&=&-\frac{2(C_2G_{13}-C_3G_{12})(C_3\kappa_2-C_2\kappa_3)}{G_{23}(C_2+C_3)^3}q^2.
\label{eq:J1}
\end{eqnarray}
$J(\mathbf{q})$ is the lowest order correction in large $G_{23}$ limit.
Finally we replace $\mathbf{q}$ by $\nabla$ and reformulate Eq.~(\ref{eq:Heff1}) 
as an effective 2TM:
\begin{eqnarray}
&&\frac{\partial}{\partial t}
\begin{bmatrix}
A_{\zeta_1}\\
A_{\zeta_+}
\end{bmatrix}
\nonumber\\
&&=
\begin{bmatrix}
-\frac{\kappa_1}{C_1}\nabla^2-\frac{G_{\rm eff}}{C_1} & \frac{G_{\rm eff}}{C_1}+J_{\zeta_1\zeta_+}\\
 \frac{G_{\rm eff}}{C_{\rm eff}} +J_{\zeta_+\zeta_1}& -\frac{\kappa_{\rm eff}}{C_{\rm eff}}\nabla^2-\frac{G_{\rm eff}}{C_{\rm eff}}+J_{\zeta_+\zeta_+}
\end{bmatrix}
\begin{bmatrix}
A_{\zeta_1}\\
A_{\zeta_+}
\end{bmatrix},\nonumber\\
\end{eqnarray}
with the effective parameters $G_{\rm eff}$, $C_{\rm eff}$, and $\kappa_{\rm eff}$ given by,
\begin{eqnarray}
G_{\rm eff} &=& G_{12}+G_{13}\\
C_{\rm eff} &=& C_2+C_3\\
\kappa_{\rm eff} &=& \kappa_2+\kappa_3,
\end{eqnarray}
and the lowest order correction terms,
\begin{eqnarray}
J_{\zeta_1\zeta_+}&=&\frac{\Omega_{13}}{G_{23}}\frac{C_2\kappa_3-C_3\kappa_2}{C_2+C_3}\nabla^2\\
J_{\zeta_+\zeta_1}&=&\frac{\{(C_3-C_2)G_{13}+2C_3G_{12}\}(C_3\kappa_2-C_2\kappa_3)}{G_{23}(C_2+C_3)^3}\nabla^2\nonumber\\
\\
J_{\zeta_+\zeta_+}&=&\frac{2(C_2G_{13}-C_3G_{12})(C_3\kappa_2-C_2\kappa_3)}{G_{23}(C_2+C_3)^3}\nabla^2
\label{eq:J2}
\end{eqnarray}
 {The appearance of the $\nabla^2$ dependent correction terms owes to the deviation of the $\zeta_+(\mathbf{q})$ branch from a parabolic dispersion at large $\mathbf{q}$. }

\section{spectrum of linear multiple temperature model}
\label{theorem}

 {It is straightforward to extend the 2TM Eq.~\eqref{eq:2TM} and 3TM Eqs.~\eqref{eq:3TM}-\eqref{eq:3TMHij} to a general $N$-temperature model.
We call it the linear multiple temperature model (MTM).
This extension is done just by increasing the number of subsystems in Eqs.~\eqref{eq:3TM}-\eqref{eq:3TMHij} from three to $N$,
The MTM is then defined by:
\begin{eqnarray}
\Lambda \frac{\partial}{\partial t}\mathbf{T}=H \mathbf{T},
\label{eq:MTM}
\end{eqnarray}
\begin{eqnarray}
\Lambda_{ij}=C_i\delta_{ij},
\end{eqnarray}
\begin{eqnarray}
H_{ii}&=&\kappa_i\nabla^2-\sum_{j\neq i}^NG_{ij},\label{eq:HNii}\\
H_{ij}&=&H_{ji}=G_{ij} \textrm{\;for\;}(i\neq j),\label{eq:HNij}
\end{eqnarray}
where $\mathbf{T}$ now denotes the $N$-components vector representing the subsystem temperatures, and the subscripts $i,j$ run from 1 to $N$.
Examples of the MTM include the nonthermal lattice model\cite{WBE16} or just multitemperature model\cite{LVCR16}, which assign phonon mode resolved temperatures.
}
In the previous sections we have found that the spectra of linear 2TM is always negative real valued, or exactly zero at $\mathbf{q}=(0,0,0)$ point.
The linear 3TM shows the same property within the parameter range we plot in Fig.~\ref{fig:3TM}(a).  {Here we prove that this physically reasonable property holds for any $N$, assuring that the temperatures of all the subsystems asymptotically tend to a common, spatially uniform, final value.}

\begin{thm}
\label{thm:negativity}
When the boundary condition Eq.(\ref{eq:BC}) is given, the linear MTM Eq.~(\ref{eq:MTM}) is transformed as
\begin{eqnarray}
\frac{\partial}{\partial t}\mathbf{T}=\Lambda^{-1}H \mathbf{T}.
\end{eqnarray}
Once the initial condition is given, we can completely determine the MTM dynamics from the eigenvalue of a matrix 
\begin{eqnarray}
H'= \Lambda^{-1}H.
\label{eq:Hp2}
\end{eqnarray}
$H'$ can be diagonalized and its eigenvalues  {$\zeta (\mathbf{q})$}
satisfy the following two properties:
\begin{enumerate}
\item the eigenvalue $\zeta (\mathbf{q})$ of  {matrix $H^\prime$} always satisfies $\zeta(\mathbf{q})\in \mathbb{R}$ and $\zeta(\mathbf{q}) \le 0$,
\item When $\zeta(\mathbf{q})=0$, $\mathbf{q}$ always satisfies $\mathbf{q} = 0$.
\end{enumerate}
\end{thm}
\begin{proof}
We consider the following eingenvalue equation
\begin{eqnarray}
\Lambda^{-1}H \mathbf{v}=\zeta \mathbf{v},
\end{eqnarray}
where $\mathbf{v}={}^t(v_1,v_2,\cdots,v_N)$ is a right eigenvector and $\zeta$ is a corresponding eigenvalue.
 {By multiplying both sides by a diagonal matrix $\Lambda^{1/2}$ from the left}, which satisfies $(\Lambda^{1/2})^2=\Lambda$, we obtain
\begin{eqnarray}
\Lambda^{-1/2}H\Lambda^{-1/2}\cdot \Lambda^{1/2}\mathbf{v}=\zeta \Lambda^{1/2}\mathbf{v}.
\end{eqnarray}
Thus $\Lambda^{1/2}\mathbf{v}$ becomes an eigenvector of a symmetric matrix $\Lambda^{-1/2}H\Lambda^{-1/2}$ whose eigenvalue is given by $\zeta$.
Clearly the $\zeta$ always satisfies $\zeta \in \mathbb{R}$.

We can further restrict the distribution of eigenvalues on the complex plain by using the Gershgorin's theorem \cite{G31}, which states that the eigenvalues of $N\times N$ matrix $A$ exist on a closed region $D$ which is defined by 
\begin{eqnarray}
D\equiv \tilde{C}_1\cup \tilde{C}_2 \cup \cdots \cup \tilde{C}_N,
\end{eqnarray}
where $\tilde{C}_i (i=1,\cdots,N)$ is a closed disk whose center position is given by $A_{ii}$ on the complex plane and its radius $R_i$ is given by
\begin{eqnarray}
R_i=\sum_{j\neq i}^N |A_{ij}|.
\end{eqnarray}
In our case the center position of the closed disc $\tilde{C}_i$ is given by
\begin{eqnarray}
\{\Lambda^{-1}H\}_{ii}= -\frac{\kappa_i}{C_i}\mathbf{q}^2-\sum_{j\neq i}\frac{G_{ij}}{C_i}=-\frac{\kappa_i}{C_i}\mathbf{q}^2-R_i,
\end{eqnarray}
and the radius $R_i$ of $\tilde{C}_i$ is 
\begin{eqnarray}
R_i=\sum_{j\neq i}\left|\frac{G_{ij}}{C_i}\right|=\sum_{j\neq i}\frac{G_{ij}}{C_i}
\end{eqnarray}
since $G_i$ and $C_i$ are positive real valued parameters.
The closed region $D$ therefore extends over a semi-infinite plain whose real part is negative, and $D$ can include the origin of complex plain only if $\mathbf{q}=0$.
We therefore conclude that $\zeta(\mathbf{q})$ is always  {non-positive} real valued and  {can be 0} only if $\mathbf{q}=0$.
\end{proof} 
 
This theorem strongly restricts the behavior of linear MTM.
For any given initial condition, the linear MTM only provides damping solutions regardless of material parameters. 
Consequently, an external heat, or maybe nonlinearity is needed to excite oscillatory and amplifying behavior in its dynamics.

In addition, we can show the following property of the MTM's solution at $\mathbf{q} = 0$.
\begin{thm}
\label{thm:q0}
When $\mathbf{q} = 0$, at least one eigenvalue of Eq.~\eqref{eq:Hp2} becomes 0 {, and at least one right eigenvector of such solutions has all its components equal.}
\begin{proof}
We firstly prove the first half of the theorem. 
The matrix $H$ given by
\begin{eqnarray}
H_{ii}&=&-\kappa_iq^2-\sum_{j\neq i}^NG_{ij},\label{eq:HNii2}\\
H_{ij}&=&H_{ji}=G_{ij} \textrm{\;for\;}(i\neq j),\label{eq:HNij2}
\end{eqnarray}
becomes linear dependent when $\mathbf{q} = 0$, \textit{i}.\textit{e}., $\det{H} = 0$ at $\mathbf{q} = 0$.
The matrix $H'=\Lambda^{-1}H$ then becomes linear dependent:   $\det{H'}=\det{\Lambda^{-1}}\det{H} = 0$ at $\mathbf{q} = 0$ limit, as well.
This implies a condition 
\begin{eqnarray}
\det{H'(\mathbf{q})}=\prod_i^N \zeta_i(\mathbf{q}) = 0 \; \textrm{for}\; \mathbf{q} = 0. 
\label{eq:Hplim}
\end{eqnarray}
Here $\zeta_i(\mathbf{q})$ $(i=1,\cdots,N)$ is an eigenvalue of $H'$.
To satisfy Eq.~(\ref{eq:Hplim}), at least one $\zeta_i(\mathbf{q})$ must fulfill a condition 
\begin{eqnarray}
\zeta_i(\mathbf{q}) = 0 \;\textrm{for}\; \mathbf{q} = 0. 
\end{eqnarray}

 {We secondly prove the latter half of the theorem.
By substituting a right eigenvector with equal components:
\begin{eqnarray}
\mathbf{v}^R={}^t[v,v,\cdots,v],
\end{eqnarray}
we can show the following:
\begin{eqnarray}
\sum_{j}\{\Lambda^{-1}H\}_{ij}v^R_j=0\;\textrm{for}\; \mathbf{q} = 0.
\end{eqnarray}
}
\end{proof}
\end{thm}
It should also be noted that Theorems 1 and 2 jointly assures that the temperatures of all the subsystems approach to a common, spatially uniform, finite value in the long time limit.  

 {From a perspective of the total energy conservation, the stationary solution $\zeta(\mathbf{q}=0)=0$ plays a special role.
The total energy $U_{\mathrm{total}}(t)$ of the system at time $t$ is given by
\begin{eqnarray}
U_{\mathrm{total}}(t)=\int_V d^3\mathbf{r} \sum_i^N C_i T_i(\mathbf{r},t)
\label{eq:Utotal}
\end{eqnarray}
where the sum runs over all the susbsystems, and the integral is taken over the system volume $V$.
Since the temperature distribution $T_i(\mathbf{r},t)$ can be decomposed into the contribution of each eigenmode, we can define the eigenmode resolved energy $U_{\zeta\mathbf{q}}(t)$ by
\begin{eqnarray}
U_{\zeta\mathbf{q}} =\int_V d^3\mathbf{r}\sum_i^N C_i v^R_{\zeta\mathbf{q},i}(\mathbf{r},t) 
\label{eq:Upartial}
\end{eqnarray}
where $v^R_{\zeta\mathbf{q},i}$ denotes the $i$-th component of the right eigenvector of the MTM Eq.~(\ref{eq:Hp2}).
The $U_{\mathrm{total}}(t)$ is retrieved by summing up $U_{\zeta\mathbf{q}}$ over the mode index $\zeta$ and wave vector $\mathbf{q}$:
\begin{eqnarray}
U_{\mathrm{total}}(t) = \sum_{\zeta,\mathbf{q}} U_{\zeta\mathbf{q}}.
\end{eqnarray}
In the absence of the external heat source, the total energy of the system Eq.~(\ref{eq:Utotal}) must be conserved.
The following theorem guarantees this requirement and, moreover, shows that all the eigenmodes but $\zeta(\mathbf{q}=0)=0$ do not hold net energy.
\begin{thm}
Given the boundary condition Eq.(\ref{eq:BC}), 
the eigenmode-resolved energy Eq.~(\ref{eq:Upartial}) vanishes:
\begin{eqnarray}
U_{\zeta\mathbf{q}}=0,
\end{eqnarray}
except for the spatially uniform, stationary mode $U_{\zeta=0,\mathbf{q}=0}$ whose eigenvalue is $\zeta(\mathbf{q}=0)=0$.
\begin{proof}
The case of $\mathbf{q}\neq 0$ is trivial since its spatial dependence is sinusoidal, which becomes zero when integrated over the system volume. 
We next show the case of $\zeta(\mathbf{q}=0)\neq 0$.
We can then write down the MTM Eq.~(\ref{eq:MTM}) as
\begin{eqnarray}
\Lambda\frac{\partial}{\partial t}\mathbf{v}^R_{\zeta,\mathbf{q}}=\Lambda \zeta(\mathbf{q})\mathbf{v}^R_{\zeta,\mathbf{q}}=H\mathbf{v}^R_{\zeta,\mathbf{q}}.
\end{eqnarray}
By using the property of $H$ for $\mathbf{q}=0$:
\begin{eqnarray}
\sum_iH_{ij}=\sum_jH_{ij}=0,
\end{eqnarray}
we can derive the following:
\begin{eqnarray}
\zeta(\mathbf{q})\sum_i\Lambda_{ii}v^R_{\zeta\mathbf{q},i}&=&
\zeta(\mathbf{q})\sum_iC_{i}v^R_{\zeta\mathbf{q},i}\nonumber\\
&=&\zeta(\mathbf{q}) U_{\zeta\mathbf{q}}\nonumber\\
&=&\sum_j(\sum_i H_{ij}) v^R_{\zeta\mathbf{q},j}=0.
\end{eqnarray}
Since $\zeta(\mathbf{q}=0)\neq 0$, $U_{\zeta\mathbf{q}}$ must be zero.
The total energy Eq.~(\ref{eq:Utotal}), therefore, has a finite contribution only from the stationary mode $\zeta=0, \mathbf{q}=0$. 
\end{proof}
\end{thm}
}
We finally show a subsidiary theorem about  {the monotonically decreasing property of the eigenvalue $\zeta(\mathbf{q})$ with respect to the magnitude $q$ of wave number.}
\begin{thm}
When the boundary condition Eq.(\ref{eq:BC}) is given, the eigenvalue $\zeta(\mathbf{q})$ of the linear MTM defined by Eq.~(\ref{eq:MTM}) always satisfies
\begin{eqnarray}
\frac{\partial}{\partial q}\zeta(\mathbf{q}) \le 0,
\end{eqnarray}
where $q=|\mathbf{q}|$.
\begin{proof}
It is sufficient to proove it for the symmetric matrix $\tilde{H}=\Lambda^{-1/2}H\Lambda^{-1/2}$ as it possesses same eigenvalues with the MTM's matrix $\Lambda^{-1}H$.
In this case the right eigenvector $\mathbf{v}$ coincides with left one.
We can then immediately write down as follows:
\begin{eqnarray}
&&\frac{\partial}{\partial q}\zeta(\mathbf{q})=
\frac{\partial}{\partial q}\sum_{ij} v_i \tilde{H}_{ij}v_j\nonumber\\
&&=\sum_{ij}\{ (\frac{\partial}{\partial q} v_i)\tilde{H}_{ij}v_j+v_i\tilde{H}_{ij}(\frac{\partial}{\partial q} v_j)
+v_i(\frac{\partial}{\partial q} H_{ij})v_j\}\nonumber \\
&&=2\zeta(\mathbf{q})\sum_i (\frac{\partial}{\partial q} v_i)v_i+\sum_iv_i^2(\frac{\partial}{\partial q}\tilde{H}_{ii})\nonumber\\
&&=-\sum_iv_i^2\frac{\kappa_i}{C_i}q \le 0.
\label{eq:thm2_proof}
\end{eqnarray}
In the third line of Eq.~(\ref{eq:thm2_proof}) we assumed the norm conservation of $\mathbf{v}$:
\begin{eqnarray}
\frac{\partial}{\partial q}\sum_iv_i^2=0.
\end{eqnarray}
\end{proof}
\end{thm}
 {This theorem physically states that the larger the the wave vector $q$ of the mode, the shorter the mode lifetime $\tau(\mathbf{q})=-1/\zeta(\mathbf{q})$, or equivalently, the faster the mode damps.}

\section{Conclusions}
\label{sec:conclusions}

{We have presented the exact analytical solutions of the linear multiple temperature model under the insulated boundary condition. By extending the familiar Fourier series expansion of the heat equation, we have shown that the system dynamics is expressed as a linear combination of eigenmodes with different wave numbers and lifetimes (or decay constants). }

{
The eigenmode picture has enabled us to unveil the hierarchical structure of the MTM that an $N+1$-temperature model approximates an $N$-temperature model in certain limits.
For example, in the small wave number $\mathbf{q}$ and long time limits, the upper branch $\zeta_+(\mathbf{q})$ of the 2TM approximates the 1TM. 
If the weak coupling limit is additionally taken, the 3TM approximates the 2TM.
We note that such approximations break in a small spacial scale.
In the opposite, large wave number limit, on the other hand, the diffusion term dominates the dynamics and the electron-lattice coupling is negligible. 
}

{
We have also found that a unique diffusion mode appears when the 3TM parameter set is asymmetric, e.g., two of the three subsystems have vanishing thermal diffusion constants $\kappa_i$.
Such an approximation is commonly used to model the electron-spin-phonon system\cite{BMDB96,KKWHAC14} or the electron-multiple phonon mode system \cite{WBE16,BOMMSR20}.
Then appear modes purely composed of the phonon-phonon or spin-phonon temperature, with the electron temperature completely suppressed, which have no analog in the 2TM.
Similar modes are expected also in four or more temperature models.}

{We have demonstrated that the linear MTM possesses physically reasonable properties.
Each eigenmode neither grows nor oscillates, but exponentially decays, except for a spatially uniform, stationary mode corresponding to the final state where all the subsystems has the same temperature.
The larger the mode wave vector, the shorter its lifetime.
Only the stationary mode has a net finite energy, assuring the conservation of the total energy.}

{The present linear MTM can be extended to better describe real materials by including external fields, non-linearity, and spatial non-uniformity etc.}
{Even in such a model, the eigenmode picture will be a useful tool to interpret the system dynamics by introducing mode excitation, mode-mode interaction, and mode scattering processes, as is done in quantum mechanics and nonlinear optics.
}

\red{}

\begin{acknowledgements}
This research was supported by MEXT
Quantum Leap Flagship Program (MEXT Q-LEAP) Grant
No. JPMXS0118067246.
\end{acknowledgements}

\bibliography{reference}
\end{document}